\documentclass[twoside,usletter]{article}
\usepackage{amsfonts}
\usepackage{amsmath}
\usepackage{amssymb}
\usepackage{qic}
\usepackage{url}
\usepackage{theory}

\newcommand{\proofskip}{2mm}
\numberwithin{equation}{section}

\begin{document}
\setlength{\textheight}{8.0truein}    

\runninghead{On Quantum-Classical Equivalence
for Composed Communication Problems}
            {Alexander~A.~Sherstov}

\normalsize\textlineskip
\thispagestyle{empty}
\setcounter{page}{1}

\copyrightheading{0}{0}{2010}{000--000}

\vspace*{0.88truein}

\alphfootnote

\fpage{1}

\centerline{\bf ON QUANTUM-CLASSICAL EQUIVALENCE}
\vspace*{0.035truein}
\centerline{\bf FOR COMPOSED COMMUNICATION PROBLEMS}

\vspace*{0.30truein}
\centerline{\footnotesize
ALEXANDER~A.~SHERSTOV}
\centerline{\footnotesize\it 
Department of Computer Sciences,
University of Texas at Austin
}
\baselineskip=10pt
\centerline{\footnotesize\it 
Austin, Texas 78757 USA
}
\vspace*{10pt}
\publisher{October 28, 2009}{January 6, 2010}

\vspace*{0.21truein}

\abstracts{
An open problem in communication complexity proposed by
several authors is to prove that for every
Boolean function $f,$ the task of computing $f(x\wedge y)$ has
polynomially related classical and quantum bounded-error complexities.
We solve a variant of this question.  For every $f,$ we
prove that the task of computing, on input \mbox{$x$ and $y,$}
\emph{both} of the quantities $f(x\wedge y)$ and $f(x\vee y)$ has
polynomially related classical and quantum bounded-error complexities.
We further show that the quantum bounded-error complexity is
polynomially related to the classical deterministic complexity
and the block sensitivity of $f.$ This result
holds regardless of prior entanglement.}{}{}

\vspace*{10pt}

\vspace*{3pt}

\vspace*{1pt}\textlineskip    
\section{Introduction \label{sec:intro}}

Quantum communication complexity, introduced by Yao~\cite{yao93quantum},
studies the amount of quantum communication necessary to compute a
Boolean function $F$ whose arguments are distributed among several
parties.  In the canonical setting, one considers a function
$F\colon X\times Y\to\zoo,$ where $X$ and $Y$ are some finite sets. One
of the parties, Alice, receives an input $x\in X,$ and the other
party, Bob, receives an input $y\in Y.$ Their objective is to
evaluate $F(x,y).$ To this end, Alice and Bob can exchange messages
through a shared quantum communication channel.
They can additionally take advantage of arbitrary \emph{prior
entanglement.} The cost of a communication protocol is the total number of qubits
exchanged in the worst case on any input $(x,y).$ The bounded-error
quantum communication complexity of $F$ with prior entanglement,
denoted $Q^*_{1/3}(F),$ is the least cost of a protocol that computes
$F$ correctly with probability at least $2/3$ on every input.
Quantum communication has an obvious classical counterpart, the
\emph{randomized model}, in which the parties exchange classical
bits ($0$ and $1$) and additionally share an unlimited supply of
unbiased random bits.  The bounded-error classical communication
complexity of $F,$ denoted $R_{1/3}(F),$ is the least cost of a
randomized protocol that computes $F$ correctly with probability
at least $2/3$ on every input.

A central goal of the field is to determine whether quantum
communication can be significantly more powerful than classical
communication, i.e., whether a superpolynomial gap exists between
the quantities $Q^*_{1/3}(F)$ and $R_{1/3}(F)$ for some function
$F\colon X\times Y\to\zoo.$ Exponential separations between quantum and
classical complexity are well known in several alternate models of
communication~\cite{ambainis03sampling,      
raz99quantum-classical,                      
bcww01fingerprinting, 
baryossef-et-al04quantum-classical,          
gavinsky-et-al06quantum-classical,           
gavinsky06fingerprinting,                    
gkkrw07quantum-classical,                    
gavinsky08quantum-classical,                 
gavinsky-pudlak08quantum-classical-oneway},  
such as one-way
communication, simultaneous message passing, sampling, and computing
a partial function or relation. However, these results do not 
apply to the original question about $Q_{1/3}^*(F)$ and
$R_{1/3}(F),$ and the largest known separation between the two
quantities is the quadratic gap for the disjointness
function~\cite{razborov03quantum,aaronson-ambainis05or}.

It is conjectured that $Q^*_{1/3}(F)$ and $R_{1/3}(F)$ are
polynomially related for all $F\colon X\times Y\to\zoo.$ Despite consistent
research efforts, this conjecture appears to be beyond
the reach of the current techniques.  An intermediate
goal, proposed by several authors~\cite{buhrman-dewolf01polynomials,
klauck01quantum-journal,shi-zhu07block-composed,shi07approx-survey}
and still unattained, is to prove the conjecture for the 
class of communication problems $F\colon\zoon\times\zoon\to\zoo$ of the
form
\[ F(x,y) = f(x\wedge y) \]
for an arbitrary function $f\colon\zoon\to\zoo.$ There has 
been encouraging progress on this problem.  
In a breakthrough result, Razborov~\cite{razborov03quantum} solved
it for the special case of symmetric $f.$ Using
unrelated techniques, 
a polynomial relationship between quantum and classical
complexity was proved in~\cite{sherstov07quantum} for the broader class of problems
$F\colon\zoo^{4n}\times\zoo^{4n}\to\zoo$ given by
\[F(x,y) = f(\dots,\; (x_{i,1}y_{i,1}
\vee \cdots \vee 
x_{i,4}y_{i,4}),\; \dots) \]
for an arbitrary function $f\colon\zoon\to\zoo.$
Independently, Shi and Zhu~\cite{shi-zhu07block-composed} used
a different approach to prove a polynomial relationship between
quantum and classical communication complexity for the family of functions
$F\colon\zoo^{kn}\times\zoo^{kn}\to\zoo$ given by
\[ 
F(x,y) = f(\dots,\g(x_{i,1},y_{i,1},\dots, x_{i,k},y_{i,k}),\dots),
\]
where $f\colon\zoon\to\zoo$ is arbitrary and $\g$ is any gadget on
$2k\geq\Omega(\log n)$ variables that has certain pseudorandom analytic
properties.
More recently, Montanaro and
Osborne~\cite{montanaro-osborne09xor} studied
quantum-classical equivalence for functions of the form
$f(x\oplus y),$ where the combining function $f$ obeys
certain constraints such as monotonicity or suitable
Fourier structure.

\subsection{Our Results} 

While the above results give further evidence that quantum
and classical communication complexities are polynomially related,
it remains open to prove this conjecture for all
functions of the form $F(x,y)=f(x\wedge y).$ 
In this paper, we solve a variant of this
question. Specifically, we consider the communication
problem of computing, on input $x,y\in\zoon,$ \emph{both} of the
quantities $f(x\wedge y)$ and $f(x\vee y).$  Our main result is a
polynomial relationship between the quantum and classical complexity
of any such problem, regardless of $f.$ We further show that the
quantum complexity of any such problem is polynomially related to
its \emph{deterministic} classical complexity $D(F)$ and
to the \emph{block sensitivity} $\bs(f)$ of $f.$ A formal definition
of block sensitivity, a well-studied combinatorial complexity measure,
will be given later in Section~\ref{sec:s-bs-dt}.

\begin{theorem}[On quantum-classical equivalence]
Let $f\colon\zoon\to\zoo$ be arbitrary. Let $F$ denote the communication
problem of computing, on input $x,y\in\zoon,$ both of the quantities
$f(x\wedge y)$ and $f(x\vee y).$ Then
\begin{align*}
D(F) \geq R_{1/3}(F) \geq Q^*_{1/3}(F)\geq 
     \Omega(\bs(f)^{1/4})\geq \Omega(D(F)^{1/12}).
\end{align*}
\label{thm:main}
\end{theorem}

A corollary of Theorem~\ref{thm:main} is that given any $f,$ a
polynomial relationship between the classical and quantum complexities
is assured for at least one of the communication problems $f(x\wedge
y),$ $f(x\vee y).$ More precisely, we have:

\begin{corollary}
Let $f\colon\zoon\to\zoo$ be arbitrary. Let $F_1$ and $F_2$ denote the
communication problems of computing $f(x\wedge y)$ and $f(x\vee
y),$ respectively. Then either
\begin{align}
D(F_1) \geq R_{1/3}(F_1) \geq Q^*_{1/3}(F_1)\geq 
     \Omega(\bs(f)^{1/4})\geq \Omega(D(F_1)^{1/12})
	 \label{eqn:F1bigger}
\end{align}
or
\begin{align}
D(F_2) \geq R_{1/3}(F_2) \geq Q^*_{1/3}(F_2)\geq 
     \Omega(\bs(f)^{1/4})\geq \Omega(D(F_2)^{1/12})
	 \label{eqn:F2bigger}
\end{align}
or both.
\end{corollary}

\begin{proof}
{
Theorem~\ref{thm:main} implies (\ref{eqn:F1bigger}) if
$Q^*_{1/3}(F_1)\geq Q^*_{1/3}(F_2)$ and implies (\ref{eqn:F2bigger})
otherwise.
}
\end{proof}

\begin{remark}
{\rm 
As a matter of formalism, the communication problem in
Theorem~\ref{thm:main} can be expressed in standard
form $F\colon X\times Y\to\zoo$ by introducing an
additional bit $b\in\zoo$ to indicate the desired
output, i.e., $f(x\wedge y)$ or $f(x\vee y).$ 
}
\end{remark}

Apart from giving a polynomial relationship between the quantum and
classical complexity of our functions, Theorem~\ref{thm:main}
shows that prior entanglement does not affect their
quantum complexity by more than a polynomial. It is an
open problem~\cite{buhrman-dewolf01polynomials} to prove a polynomial
relationship for quantum communication complexity with and without prior
entanglement, up to an additive logarithmic term.
Known separations here are quite modest: entanglement allows
for a factor of $2$ savings via superdense coding, as
well as an additive $\Theta(\log n)$ savings for the
equality function.
Finally, we prove in Section~\ref{sec:logrank} that the communication
problems in Theorem~\ref{thm:main} satisfy another well-known conjecture, the 
\emph{log-rank conjecture} of Lov{\'a}sz and Saks~\cite{lovasz-saks88logrank}. 

Up to this point, we have focused on the communication problem of
computing \mbox{$f(x\wedge y)$} and $f(x\vee y).$ In
Section~\ref{sec:composed}, we consider quantum-classical equivalence
and the log-rank conjecture in a broader context.  Specifically,
we consider general compositions of the form
$f(\dots,\g_i(x^{(i)},y^{(i)}),\dots),$ where one has a combining
function $f\colon \zoon\to\zoo$ that receives input from intermediate
functions $\g_i\colon X_i\times Y_i\to\zoo,$ $i=1,2,\dots,n.$ We
show that under natural assumptions on $\g_1,\dots,\g_n,$ the
composed function will have polynomially related quantum and classical
bounded-error complexities and will satisfy the log-rank conjecture.

\subsection{Our Techniques}

We obtain our main result by bringing together \emph{analytic} and
\emph{combinatorial} views of the uniform approximation of
Boolean functions. The analytic approach and combinatorial approach
have each found important applications in
isolation, e.g.,~\cite{nisan-szegedy94degree,
beals-et-al01quantum-by-polynomials,   
buhrman-dewolf01polynomials,           
razborov03quantum, 
sherstov07quantum,
shi-zhu07block-composed}. The key to our work is to find a way
to combine them.

On the analytic side, a key ingredient in our solution is
the \emph{pattern matrix method}, developed in~\cite{sherstov07ac-majmaj,
sherstov07quantum}.
Let $f\colon\zoon\to\zoo$ be a given function. The pattern matrix method
centers around a communication game in which Alice is given a string
$x\in\zoo^N,$ where $N\geq 4n$; Bob is given a subset
$S\subset\{1,2,\dots,N\},$ where $|S|=n$; and their objective is
to compute $f(x|_S),$ where $x|_S=(x_{i_1},\dots,x_{i_n})\in\zoon$
and $i_1<\cdots<i_n$ are the elements of $S.$ 
The pattern matrix method gives a lower bound on
the communication complexity of this problem in a given model (e.g.,
randomized, bounded-error quantum with prior entanglement,
unbounded-error, weakly-unbounded error) in terms of the 
corresponding analytic property
of $f$ (e.g., its approximate degree or threshold degree).

Essential to the pattern matrix method,
as applied in this paper, is a closed-form expression
for the singular values of every matrix of the form
\begin{align}
\Psi=\Big[\psi(x|_S\,\oplus\, w)\Big]_{x,(S,w)} \label{eqn:Psi}
\end{align}
in terms of the Fourier spectrum of the function $\psi\colon\zoon\to\Re,$
where $x$ and $S$ are as described in the previous paragraph and
$w$ ranges over $\zoon.$ The method critically
exploits the fact that the rows of $\Psi$ are applications of the
\emph{same} function $\psi$ to various subsets of the variables or
their negations.  In the communication problems of this paper,
this assumption is violated: as Bob's input $y$ ranges
over $\zoon,$ the induced functions $f_y(x)=f(x\wedge y)$ 
may have nothing to do with each other.  This
obstacle is fundamental: allowing a distinct function $\psi$ in
each row of~(\ref{eqn:Psi}) disrupts the spectral
structure of $\Psi$ and makes it impossible to force the desired
spectral bounds.

We overcome this obstacle by exploiting the additional
\emph{combinatorial} structure of the base function $f\colon\zoon\to\zoo,$
which did not figure in previous work~\cite{sherstov07ac-majmaj,sherstov07quantum}.
Specifically, we consider the sensitivity of $f$, the block
sensitivity of $f,$ and their polynomial equivalence in our restricted setting, as proved by Kenyon
and Kutin~\cite{ell-block-sensitivity04}.  We use this combinatorial
structure to identify a large submatrix inside $[f(x\wedge y)]_{x,y}$ or
$[f(x\vee y)]_{x,y}$ which, albeit not directly representable
in the form (\ref{eqn:Psi}), has a certain \emph{dual} matrix that
can be represented precisely in this way.  Since the pattern matrix
method relies only on the spectral structure of this dual matrix,
we are able to achieve our goal and place a strong lower bound on
the quantum communication complexity. The corresponding upper bound for
classical protocols has a short proof using a well-known argument in the
literature~\cite{bcw98quantum, beals-et-al01quantum-by-polynomials,
razborov03quantum, sherstov07quantum, 
shi-zhu07block-composed}.

The above program can be equivalently described in
terms of polynomials rather than functions. Let $\mathcal{F}$ be a
subset of Boolean functions $\zoon\to\zoo$ none of which can be
approximated within $\epsilon$ in the $\ell_\infty$ norm by a
polynomial of degree less than $d.$ For each $f\in \mathcal{F},$
linear programming duality implies the existence of a function
$\psi\colon\zoon\to\Re$ such that $\sum_{x\in\zoon}
\psi(x)f(x)>\epsilon \sum_{x\in\zoon}|\psi(x)|$
and $\psi$ has zero Fourier mass on the characters of
order less than $d.$ This dual object $\psi$
witnesses the fact that $f$ has no low-degree approximant. Now,
there is no reason to believe that a \emph{single} witness
$\psi$ can be found that works for every function in
$\mathcal{F}.$ A key technical challenge in this work is to show
that, under suitable combinatorial constraints that hold in our
setting, the family $\mathcal{F}$ will indeed have a common witness
$\psi.$ In conjunction with the pattern matrix method,
we are then able to solve the original problem.
To clarify the relevance of this discussion to the
study of functions of the form $f(x\wedge y),$
the family $\mathcal{F}$ in question is the family of the
induced functions $f_y(x)=f(x\wedge y)$ 
as the input $y$ ranges over $\zoon.$

\section{Preliminaries} \label{sec:prelim}

For convenience of notation, we will view Boolean functions in the
remainder of the paper as mappings $f\colon X\to\moo$ for some finite set
$X,$ where $-1$ corresponds to ``true.'' Note that this is a departure
from the introduction, where we used the more traditional
range $\{0,1\}.$ For $x\in\zoon,$ we define $|x|=
x_1+x_2+\cdots+x_n.$ The symbol $P_d$ stands for the set of all
univariate real polynomials of degree at most $d.$ For a given
function $f\colon\zoon\to\Re$ and a string $z\in\zoon,$ we let $f_z$
stand for the function $f_z\colon\zoon\to\Re$ given by $f_z(x)\equiv
f(x\oplus z).$ For $b\in\zoo,$ we use the notation $\overline b =
1-b = 1\oplus b.$ The \emph{characteristic vector} of a set
$S\subseteq\oneton$ is the string $\1_S\in\zoon$ such that $(\1_S)_i=1$
for $i\in S,$ and $(\1_S)_i=0$ otherwise.
For a string $x\in\zoon$ and a set $S\subseteq\oneton,$ we define
$x|_{S}=(x_{i_1},x_{i_2},\dots,x_{i_{|S|}})\in\zoo^{|S|},$ where
$i_1<i_2<\cdots<i_{|S|}$ are the elements of $S.$

\subsection{Matrices}
The symbol $\Re^{m\times n}$ refers to the family of all $m\times
n$ matrices with real entries.  
We specify a matrix by its generic entry, e.g., the notation 
$A=[F(i,j)]_{i,j}$ means that the $(i,j)$th entry of $A$ is given by the
expression $F(i,j).$ In most matrices that arise in this work, the
exact ordering of the columns (and rows) is irrelevant. In such cases
we describe a matrix by the notation $[F(i,j)]_{i\in I,\, j\in J},$
where $I$ and $J$ are some index sets.

Let $A=[A_{ij}]\in\Re^{m\times n}$ be given. We adopt the shorthands 
$\|A\|_\infty = \max |A_{ij}|$ and $\|A\|_1=\sum |A_{ij}|.$ We denote the
singular values of $A$ by
$\sigma_1(A)\geq\sigma_2(A)\geq\cdots\geq\sigma_{\min\{m,n\}}(A)\geq0.$
Recall that the {spectral norm} of $A$ is given by
\begin{align*}
\|A\|= \max_{x\in\Re^n,\; \|x\|_2=1} \|Ax\|_2 = \sigma_1(A),
\end{align*}
where $\|\cdot\|_2$ is the Euclidean vector norm.
For $A,B\in\Re^{m\times n},$ we write $\langle A, B\rangle =
\sum A_{ij}B_{ij}.$
We denote the rank of $A$ over the reals by $\rk A.$ 

We will need the following formulation of linear
programming duality in matrix notation.

\begin{theorem}[Duality]
\label{thm:lp}
For $A\in\Re^{m\times n}$ and $b\in\Re^m,$ the system
$Ax\geq b$ has no solution in $x\in\Re^n$ if and only
if there is a vector $y\in[0,\infty)^m$ such that $y\tr
A=0$ but $y\tr b>0.$
\end{theorem}

\noindent
The monograph by Schrijver~\cite[Chap.~7]{lpbook}
provides detailed background on Theorem~\ref{thm:lp} and
various other formulations of linear programming
duality, along with historical notes.

\subsection{Fourier Transform}
Consider the vector space of real functions on $\zoon,$ equipped with
the inner product
\[\langle f,\g\rangle =2^{-n} \sum_{x\in\zoon}f(x)\g(x)\]
and normed by
\[ \| f \|_\infty = \max_{x\in\zoon} |f(x)|. \]
For $S\subseteq\oneton,$ define $\chi_S\colon\zoon\to\moo$ by
$\chi_S(x) =(-1)^{\sum_{i\in S} x_i}.$
Then the functions $\chi_S,$ $S\subseteq\oneton,$ form an orthonormal
basis for the inner product space in question.  As a result, every
function $f\colon\zoon\to\Re$ has a unique representation of the form
\[f=\sum_{S\subseteq\oneton} \hat f(S)\,\chi_S,\] where $\hat
f(S)=\langle f,\chi_S\rangle$ is the \emph{Fourier coefficient} of $f$
that corresponds to the character $\chi_S.$ 
%
The following bound is immediate from the definition of Fourier
coefficients:
\begin{align}
   \max_{S\subseteq\oneton} |\hat f(S)| \leq 2^{-n} \sum_{x\in\zoon}
|f(x)|. 
\label{eqn:trivial-Fourier-bound}
\end{align}

\subsection{Monomial Count, Sensitivity, and Decision Trees}
\label{sec:s-bs-dt}

Every function $f\colon\zoon\to\Re$ has a unique representation of the form
\[f(x) = \sum_{S\subseteq\{1,\dots,n\}} \alpha_S \prod_{i\in S} x_i\]
for some reals $\alpha_S.$ We define the \emph{degree} of $f$ by
$\deg(f) = \max\{ |S|: \alpha_S\ne
0\}$ and the \emph{number of monomials} in $f$ by 
$\mon(f) = |\{ S: \alpha_S\ne 0\}|.$ 

For $i=1,2,\dots,n,$ we let $e_i\in\zoon$ stand for the vector with
$1$ in the $i$th component and zeroes everywhere else.
For a set $S\subseteq\{1,\dots,n\},$ we define $e_S\in\zoon$ by 
$e_S = \sum_{i\in S} e_i.$ In particular, $e_{\varnothing}=0.$
Fix a Boolean function $f\colon\zoon\to\moo.$ For $\ell=1,2,\dots,n,$
the \emph{$\ell$-block sensitivity} of $f,$ denoted $\bs_{\ell}(f),$
is defined as the largest $k$ for which there exist 
nonempty disjoint sets
$S_1,\dots,S_k\subseteq\{1,\dots,n\},$ each containing no more than
$\ell$ elements, such that 
\[ 
   f(z\oplus e_{S_1})=
   f(z\oplus e_{S_2})=
   \cdots=
   f(z\oplus e_{S_k})\ne
   f(z)
\]
for some $z\in\zoon.$ One distinguishes two extremal cases. The
\emph{sensitivity} of $f,$ denoted $\s(f),$ is defined by $\s(f)
=\bs_1(f).$ The \emph{block sensitivity} of $f,$ denoted $\bs(f),$
is defined by $\bs(f) = \bs_n(f).$ In this context, the term
\emph{block} simply refers to a subset $S\subseteq\oneton.$ We say
that a block $S\subseteq\oneton$ is \emph{sensitive} for $f$ on input
$z$ if $f(z)\ne f(z\oplus e_{S}).$

Following Buhrman and de Wolf~\cite{buhrman-dewolf01polynomials}, we define one additional
variant of sensitivity. The \emph{zero block sensitivity} of $f,$ denoted
$\zbs(f),$ is the largest $k$ for which there exist 
nonempty disjoint sets
$S_1,\dots,S_k\subseteq\{1,\dots,n\}$ such that 
\[ 
   f(z\oplus e_{S_1})=
   f(z\oplus e_{S_2})=
   \cdots=
   f(z\oplus e_{S_k})\ne
   f(z)
\]
for some $z\in\zoon$ with $z|_{S_1\cup \cdots\cup S_k}=(0,0,\dots,0).$

For a function $f\colon\zoon\to\moo,$ we let $\dt(f)$ stand for the least depth
of a decision tree for $f.$ The following inequalities are known.

\begin{theorem}[{Beals et
al.~\cite[\S5]{beals-et-al01quantum-by-polynomials}}] 
\label{thm:beals}
Every function $f\colon\zoon\to\moo$ satisfies
\begin{align*}
\dt(f)\leq \bs(f)^3.
\end{align*}
\end{theorem}

\begin{theorem}
[Midrijanis~\cite{midrijanis04dt-versus-polynomials}]
\label{thm:dt-deg}
Every function $f\colon\zoon\to\moo$ satisfies
\begin{align*}
\dt(f)\leq O(\deg(f)^3).
\end{align*}
\end{theorem}

For further background on these combinatorial complexity measures, we refer
the reader to the excellent survey by Buhrman and de 
Wolf~\cite{buhrman-dewolf02DT-survey}.

\subsection{Symmetric Functions}
Let $S_n$ denote the symmetric group on $n$ elements. For $\sigma\in S_n$ and
$x\in\zoon$, we denote by $\sigma x$ the string
$(x_{\sigma(1)},\ldots,x_{\sigma(n)})\in\zoon.$
A function $\phi\colon\zoon\to\Re$ is called
\emph{symmetric} if $\phi(x) = \phi(\sigma x)$ for every $x\in\zoon$ and
every $\sigma\in S_n.$ Equivalently, $\phi$ is symmetric if $\phi(x)$
is uniquely determined by $|x|.$ Observe that for every $\phi\colon\zoon\to\Re$
(symmetric or not), the derived function
\[ \phi\symm(x)    =    \Exp_{\sigma\in S_n} [\phi(\sigma x)]
\]
is symmetric.     Symmetric functions on $\zoon$ are intimately
related to univariate polynomials, as demonstrated by Minsky and
Papert's \emph{symmetrization argument}~\cite{minsky88perceptrons}:

\begin{proposition}[Minsky and Papert]
Let $\phi\colon\zoon\to\Re$ be given such that $\hat \phi(S)=0$ for $|S|>r.$ 
Then there is a polynomial $p\in P_r$ with
\[  \Exp_{\sigma\in S_n} [\phi(\sigma x)] = p(|x|), 
\qquad x\in\zoon.  \]
\label{prop:symmetrization}
\end{proposition}

\subsection{Pattern Matrices} \label{sec:pattern}
\emph{Pattern matrices}, introduced 
in~\cite{sherstov07ac-majmaj, sherstov07quantum}, play an
important role in this paper.  Relevant definitions and results
from~\cite{sherstov07quantum} follow.

Let $n$ and $N$ be positive integers with $n\mid N.$ Split $[N]$
into $n$ contiguous blocks, with $N/n$ elements each:
\[
[N]= \left\{1,  2,  \dots,   \frac{N}{n}\right\}
       \cup
        \left\{\frac{N}{n}+1,   \dots,  \frac{2N}{n}\right\}
         \cup \cdots \cup
         \left\{\frac{(n-1)N}{n}+1,  \dots,  N\right\}.
\]
Let $\VV(N,n)$ denote the family of subsets $V\subseteq\{1,\dots,N\}$
that have exactly one element from each of these blocks (in particular,
$|V|=n$).  Clearly, $|\VV(N,n)|=(N/n)^n.$ 

\begin{definition}[Pattern matrix]
{\rm 
For $\phi\colon\zoo^n\to\Re,$ the \emph{$(N,n,\phi)$-pattern matrix} is
the real matrix $A$ given by
\[ A = \Big[\phi(x|_V\oplus
w)\Big]_{x\in\zoo^N,\,(V,w)\in\VV(N,n)\times\zoo^n}  \;. \]
In words, $A$ is the matrix of size $2^N$~by~$(N/n)^n2^n$ whose
rows are indexed by strings $x\in\zoo^N,$ whose columns are indexed
by pairs $(V,w)\in\VV(N,n)\times\zoo^n,$ and whose entries are given
by $A_{x,(V,w)}= \phi(x|_V\oplus w).$
}
\end{definition}

\noindent
The logic behind the term ``pattern matrix" is as follows: a mosaic
arises from repetitions of a pattern in the same way that $A$ arises
from applications of $\phi$ to various subsets of the variables.
We are going to need the following expression for the spectral norm of a pattern
matrix~\cite[Thm.~4.3]{sherstov07quantum}.

\begin{theorem}[Sherstov]
Let $\phi\colon\zoo^n\to\Re$ be given. Let $A$ be the $(N,n,\phi)$-pattern
matrix. Then 
\[ 
\|A\| \;=\; \sqrt{ 2^{N+n} \left( \frac{N}{n}\right)^n} 
\;\max_{S\subseteq\{1,\dots,n\}} 
    \left\{ |\hat\phi(S)| \left( \frac{n}{N}\right)^{|S|/2} \right\}.
\]
\label{thm:pattern-spectrum}
\end{theorem}

By identifying a set $S\subseteq\{1,2,\dots,N\}$ with its characteristic
vector $\mathbf{1}_S\in\zoo^N,$ we may alternately regard $\VV(N,n)$ as a
family of strings in $\zoo^N$ rather than as a family of sets. This view
will be useful in the proof of
Theorem~\ref{thm:main-technical-pattern-theorem} below.
Detailed background on the pattern matrix method is
available in the survey article~\cite{dual-survey}.

\subsection{Communication Complexity} \label{sec:quantum-model}

This section reviews the quantum model of communication complexity.
We include this review mainly for completeness; our proofs rely solely on
a standard matrix-analytic property of quantum protocols and on no other
aspect of quantum communication.

There are several equivalent ways to describe a quantum communication
protocol, e.g.,
\cite{buhrman2000quantum-survey,dewolf-thesis,razborov03quantum}.
Our description closely follows
Razborov~\cite{razborov03quantum}.  Let $\AA$ and $\BB$ be complex
finite-dimensional Hilbert spaces.  Let $\CC$ be a Hilbert space
of dimension $2,$ whose orthonormal basis we denote by
$|0\rangle,\;|1\rangle.$  Consider the tensor product
$\AA\otimes\CC\otimes\BB,$ which is itself a Hilbert space with an
inner product inherited from $\AA,$ $\BB,$ and $\CC.$ The \emph{state}
of a quantum system is a unit vector in $\AA\otimes\CC\otimes\BB,$
and conversely any such unit vector corresponds to a distinct quantum
state.  The quantum system starts in a given state and traverses a
sequence of states, each obtained from the previous one via a unitary
transformation chosen according to the protocol.  Formally, a
\emph{quantum communication protocol} is a finite sequence of unitary
transformations
\[U_1\otimes I_\BB,\quad
 I_\AA\otimes U_2,\quad
 U_3\otimes I_\BB,\quad
 I_\AA\otimes U_4,\quad
 \dots,\quad
U_{2k-1}\otimes I_\BB,\quad
I_\AA\otimes U_{2k},\]
where: $I_\AA$ and $I_\BB$ are the identity transformations in $\AA$
and $\BB,$ respectively; $U_1,U_3,\dots,U_{2k-1}$ are unitary
transformations in $\AA\otimes\CC$; and $U_2,U_4,\dots,U_{2k}$ are
unitary transformations in $\CC\otimes\BB.$ The \emph{cost} of the
protocol is the length of this sequence, namely, $2k.$ On Alice's
input $x\in X$ and Bob's input $y\in Y$ (where $X,Y$ are given finite sets),
the computation proceeds as follows.

\begin{enumerate}
\item The quantum system starts out in an initial state $\init(x,y).$
\item Through successive applications of the above unitary
transformations, the system reaches the state
\[
 \final(x,y)   =    (I_\AA\otimes U_{2k})
 (U_{2k-1}\otimes I_\BB)
 \cdots
 (I_\AA\otimes U_2)
 (U_1\otimes I_\BB)\;\init(x,y). \]
\item Let $v$ denote the projection of $\final(x,y)$ onto
$\AA\otimes\Span(|1\rangle) \otimes\BB.$ The output of the protocol
is $-1$ with probability $\langle v,v\rangle,$ and $+1$ with the
complementary probability
$1-\langle v,v\rangle.$
\end{enumerate}

All that remains is to specify how the initial state
$\init(x,y)\in\AA\otimes\CC\otimes\BB$ is constructed from
$x,y.$ It is here that the model with prior entanglement
differs from the model without prior entanglement.
In the model without prior entanglement, $\AA$  and $\BB$ have
orthonormal bases 
$\{|x,w\rangle: x\in X,\; w\in W \}$ and 
$\{|y,w\rangle: y\in Y,\; w\in W \},$ 
respectively, where $W$ is a finite set corresponding to the private
workspace of each of the parties. The initial state is the pure
state
\[ 
\init(x,y) = |x,0\rangle\,|0\rangle\,|y,0\rangle, 
\]
where $0\in W$ is a certain fixed element.  In the model with prior
entanglement, the spaces $\AA$ and $\BB$ have orthonormal bases
$\{|x,w,e\rangle:x\in X,\;w\in W,\; e\in E\}$ and 
$\{|y,w,e\rangle:y\in Y,\;w\in W,\; e\in E\},$
respectively, where $W$ is as before and $E$ is a finite set
corresponding to the prior entanglement.  The initial state is now 
the entangled state
\[ \init(x,y) =
\frac{1}{\sqrt{|E|}} \sum_{e\in E} 
|x,0,e\rangle\,|0\rangle\,|y,0,e\rangle.
\]
Apart from finite size, no assumptions are made about $W$ or $E.$
In particular, the model with prior entanglement allows for an
unlimited supply of entangled qubits.  This mirrors the unlimited
supply of shared random bits in the classical public-coin randomized
model.

Let $f\colon X\times Y\to\moo$ be a given function. A quantum
protocol $P$ is said to compute $f$ with error $\epsilon$ if
\[\Prob[P(x,y)\ne f(x,y)]\leq \epsilon\] for all $x,y,$
where the random variable $P(x,y)\in\moo$ is the output of the
protocol on input $(x,y).$ Let $Q_\epsilon(f)$ denote the least cost of
a quantum protocol without prior entanglement that computes
$f$ with error $\epsilon.$ Define $Q^*_\epsilon(f)$ analogously for
protocols with prior entanglement.  The precise choice of a constant
$\epsilon\in(0,1/2)$ affects $Q_\epsilon(f)$ and $Q^*_\epsilon(f)$
by at most a constant factor, and thus the setting $\epsilon=1/3$
entails no loss of generality. 
By the communication complexity of a Boolean
matrix $F=[F_{ij}]_{i\in I,\,j\in J}$ will be meant the communication
complexity of the associated function $f\colon I\times J\to\moo$ given
by $f(i,j) = F_{ij}.$

A useful technique for proving lower bounds on quantum communication
complexity, regardless of prior entanglement, is the \emph{generalized
discrepancy method,} originally applied by Klauck~\cite{klauck01quantum-journal}
and reformulated more broadly by Razborov~\cite{razborov03quantum}.
The following is an adaptation by
the author~\cite[Sec.~2.4]{sherstov07quantum}.

\begin{theorem}[Generalized discrepancy method] 
Fix finite sets $X,Y$ and a given function \mbox{$f\colon X\times Y\to\moo$}.
Let $\Psi=[\Psi_{xy}]_{x\in X,\,y\in Y}$ be any real matrix with $\|\Psi\|_1=1.$
Then for each $\epsilon>0,$
\[ 4^{Q_\epsilon(f)} \geq 4^{Q^*_\epsilon(f)} \geq
\frac{\langle \Psi, F\rangle - 2\epsilon}{3\,\|\Psi\|\sqrt{|X|\,|Y|}}, \]
where $F=[f(x,y)]_{x\in X,\,y\in Y}.$
\label{thm:discrepancy-method}
\end{theorem}

Apart from quantum communication, we will consider two classical
models.  For a function $f\colon X\times Y\to\moo,$ we let $D(f)$ stand
for the \emph{deterministic} communication complexity of $f.$ We
let $R_{1/3}(f)$ stand for the public-coin randomized communication complexity of
$f,$ with error probability at most $1/3.$ The following result of
Mehlhorn and Schmidt~\cite{mehlhorn-schmidt82rank-cc} gives a
powerful technique for proving lower bounds on deterministic
communication.

\begin{theorem}[Mehlhorn and Schmidt]
Let $f\colon X\times Y\to\moo$ be a given function, where $X,Y$ are finite sets. 
Put $F=[f(x,y)]_{x\in X,\, y\in Y}.$ Then
\begin{align*}
 D(f) \geq \log_2 \rk F. 
\end{align*}
\label{thm:mehlhorn-schmidt}
\end{theorem}

An excellent reference on classical communication
complexity is 
the monograph by Kushilevitz and Nisan~\cite{ccbook}.

\section{Combinatorial Ingredients} \label{sec:combinatorial}

In this section, we develop the combinatorial component of our
solution. 
We start by recalling an elegant result, due to Kenyon and
Kutin~\cite[Cor.~3.1]{ell-block-sensitivity04}, that the sensitivity
and $\ell$-block sensitivity of a Boolean function are polynomially
related for all constant $\ell.$ For the purposes of this paper,
the case $\ell=2$ is all that is needed.

\begin{theorem}[Kenyon and Kutin]
Let $f:\zoon\to\moo$ be given. Then
\begin{align*}
\s(f) \geq \alpha \sqrt{\bs_2(f)}
\end{align*}
for some absolute constant $\alpha>0.$
\label{thm:s-bs}
\end{theorem}
\noindent
\begin{remark}
{\rm
The lower bound in Theorem~\ref{thm:s-bs} is asymptotically tight,
by a construction due to Rubinstein~\cite{rubinstein95sensitivity}.
}
\end{remark}

For our purposes, the key consequence of Kenyon and Kutin's result is the
following lemma.

\begin{lemma}
\label{lem:projection}
Let $f\colon\zoon\to\moo$ be a given function. Then there exists
$\g\colon\zoon\to\moo$ such that
\begin{align}
\s(\g)&\geq \alpha\sqrt{\bs(f)}
\label{eqn:g-sensitivity}
\end{align}
for some absolute constant $\alpha>0$ and
\begin{align}
\g(x) \equiv f(x_{i_1},x_{i_2},\dots,x_{i_n})
\label{eqn:g-subfunction}
\end{align}
for some $i_1,i_2,\dots,i_n\in\{1,2,\dots,n\}.$
\end{lemma}

\begin{proof}
{
Put $k=\bs(f)$ and fix disjoint sets $S_1,\dots,S_k\subseteq\oneton$ such
that one has
$f(z\oplus e_{S_1})=
   f(z\oplus e_{S_2})=
   \cdots=
   f(z\oplus e_{S_k})\ne
   f(z)$
for some $z\in\zoon.$ Let $I$ be the set of all indices $i$ such
the string $z|_{S_i}$ features both zeroes and ones. Put $|I|=r.$
For convenience of notation, we will assume that $I=\{1,2,\dots,r\}.$
For $i=1,2,\dots,r,$ form the partition $S_i=A_i\cup B_i,$ where
\begin{align*}
 A_i = \{j\in S_i: z_j=0\}, \qquad
 B_i = \{j\in S_i: z_j=1\}. 
\end{align*}
Now let
\begin{align*}
\g(x) = f\left(\bigoplus_{i=1}^r    x_{\min A_i} e_{A_i} \oplus
              \bigoplus_{i=1}^r    x_{\min B_i} e_{B_i} \oplus 
              \bigoplus_{i=r+1}^k  x_{\min S_i} e_{S_i} \oplus 
		      \bigoplus_{i\notin S_1\cup\cdots\cup S_k}
				 x_ie_i\right).
\end{align*}
Then (\ref{eqn:g-subfunction}) is immediate.  By the properties of
$f,$ we have $\bs_2(\g)\geq k,$ with the blocks 
$\{\min A_1,\min B_1\},
  \dots,
 \{\min A_r,\min B_r\}$ 
and $\{\min S_{r+1}\},\dots,\{\min S_k\}$ being sensitive for $\g$
on input $x=z.$ As a result, Theorem~\ref{thm:s-bs} implies
(\ref{eqn:g-sensitivity}).
}
\end{proof}

\section{Analytic Ingredients} \label{eqn:analytic}

We now turn to the analytic component of our solution.  The main
results of this section can all be derived by modifying
Razborov's proof of the quantum lower bound for the disjointness
function~\cite{razborov03quantum}.  The alternate derivation
presented here has some advantages, as we
discuss in Remark~\ref{rem:qce-advantages-over-razborov}.  We
start by exhibiting a large family of Boolean functions whose
inapproximability by low-degree polynomials in the uniform norm
can be witnessed by a single, common dual object.

\begin{theorem}
\label{thm:duality}
Let $\Fcal$ denote the set of all functions $f\colon\zoon\to\moo$ such that
$f(e_1)=f(e_2)=\cdots=f(e_n)\ne f(0)=1.$ Let $\delta>0$
be a sufficiently small 
absolute constant. Then there exists a function $\psi\colon\zoon\to\Re$
such that:
\begin{align}
&\;\,\hat\psi(S)=0, && |S|<\delta\sqrt n,
   \label{eqn:mask-psi-orthog}\\
&\rule{0mm}{6mm}\sum_{x\in\zoon}|\psi(x)| =1,
   \label{eqn:mask-psi-bounded}\\
&\sum_{x\in\zoon}\psi(x)f(x) > \frac13, && f\in \Fcal.
   \label{eqn:mask-psi-correl}
\end{align}
\end{theorem}

\begin{proof}
{
Let $p$ be a univariate real polynomial that satisfies
\begin{align*}
p(0)&\in[2/3,4/3],\\
p(1)&\in[-4/3,-2/3],\\
p(i)&\in[-4/3,4/3], && i=2,3,\dots,n.
\end{align*}
It follows from basic approximation theory (viz.,
the inequalities due to A.~A.~Markov and S.~N.~Bernstein) that any such
polynomial $p$ has degree at least $\delta \sqrt n$ for an 
absolute constant $\delta>0.$ See Nisan and
Szegedy~\cite{nisan-szegedy94degree}, pp.~308--309, for a short
derivation.

By the symmetrization argument (Proposition~\ref{prop:symmetrization}),
there does not exist a multivariate polynomial $\phi(x_1,\dots,x_n)$
of degree less than $\delta\sqrt n$ such that
\begin{align*}
\phi(0)&\in[2/3,4/3],\\
\phi(e_i)&\in[-4/3,-2/3],  && i=1,2,\dots,n,\\
\phi(x)&\in[-4/3,4/3], && x\in\zoon\setminus\{0,e_1,e_2,\dots,e_n\}.
\end{align*}
Equivalently, the following system of linear
constraints has no solution in the reals $\alpha_S$:
\begin{align*}
\sum_{|S|<\delta\sqrt n} \alpha_S\chi_S(0)&\in[2/3,4/3],\\
\sum_{|S|<\delta\sqrt n} \alpha_S\chi_S(e_i)&\in[-4/3,-2/3],  && i=1,2,\dots,n,\\
\sum_{|S|<\delta\sqrt n} \alpha_S\chi_S(x)&\in[-4/3,4/3], && x\in\zoon\setminus\{0,e_1,e_2,\dots,e_n\}.
\end{align*}
The duality of linear programming
(Theorem~\ref{thm:lp}) now implies the existence of $\psi$ that
obeys
(\ref{eqn:mask-psi-orthog}), (\ref{eqn:mask-psi-bounded}), and additionally satisfies
\begin{align*}
   \psi(0) -\sum_{i=1}^n \psi(e_i) - 
   \sum_{\substack{x\in\zoon\\ |x|\geq 2}} |\psi(x)| > \frac13,
\end{align*}
which forces (\ref{eqn:mask-psi-correl}).
}
\end{proof}

\vspace{\proofskip}

We are now in a position to prove our main technical criterion for
high quantum communication complexity. Our proof is based on the
pattern matrix method~\cite{sherstov07ac-majmaj,
sherstov07quantum}. 
The novelty of the development below resides in allowing the rows of
the given Boolean matrix to derive from distinct Boolean functions,
which considerably disrupts the spectral structure.  We are able
to force the same quantitative conclusion by using the fact that
these Boolean functions, albeit distinct, share the relevant
dual object.  

\begin{theorem}
  \label{thm:main-technical-pattern-theorem}
Let $\g\colon\zoon\to\moo$ be a function such that 
$\g(z\oplus e_1)=
 \g(z\oplus e_2)=
 \cdots=
 \g(z\oplus e_k)
 \ne \g(z)$ 
for some $z\in\zoon$ with $z_1=\cdots=z_k=0.$ Then the matrix
$G=[\g(x\wedge y)]_{x,y\in\zoon}$ satisfies
\begin{align*}
	Q^*_{1/3}(G) \geq \Omega(\sqrt k).
\end{align*}
\end{theorem}

\begin{remark}
\label{rem:qce-advantages-over-razborov}
{\rm 
As formulated above,
Theorem~\ref{thm:main-technical-pattern-theorem} can be 
derived by modifying Razborov's proof of the $\Omega(\sqrt n)$
quantum lower bound for the disjointness
function~\cite[\S5.3]{razborov03quantum}.  The derivation that we
are about to give offers some advantages.  First, it is
simpler and in particular does not require tools such as Hahn
matrices in~\cite{razborov03quantum}. Second, it generalizes to
any family $\Fcal$ of functions with a common dual polynomial,
whereas the method in~\cite{razborov03quantum} is restricted to
symmetrizable families.  
}
\end{remark}

\indent
{\bf Proof (of Theorem~\textup{\ref{thm:main-technical-pattern-theorem}})}
{
Without loss of generality, we may assume that $k$ is divisible by $4.$
Let $\Fcal$ denote the system of all functions $f\colon\zoo^{k/4}\to\moo$
such that $f(e_1)=f(e_2)=\cdots=f(e_{k/4})\ne f(0)=1.$ By
Theorem~\ref{thm:duality}, there exists $\psi\colon\zoo^{k/4}\to\Re$
such that
\begin{align}
&\;\,\hat\psi(S)=0, && |S|<\delta\sqrt k,
   \label{eqn:mask-psi-orthog-restated}\\
&\rule{0mm}{6mm}\sum_{x\in\zoo^{k/4}}|\psi(x)| =1,
   \label{eqn:mask-psi-bounded-restated}\\
&\sum_{x\in\zoo^{k/4}}\psi(x)f(x) > \frac13, && f\in \Fcal,
   \label{eqn:mask-psi-correl-restated}
\end{align}
where $\delta>0$ is an absolute constant.  Now, let $\Psi$ be the
$(k/2,k/4,2^{-3k/4}\psi)$-pattern matrix.  It follows from
(\ref{eqn:mask-psi-bounded-restated}) that
\begin{align}
\|\Psi\|_1 = 1.
   \label{eqn:Psi-bounded}
\end{align}
By (\ref{eqn:trivial-Fourier-bound}) and (\ref{eqn:mask-psi-bounded-restated}),
\begin{align}
\max_{S} |\hat\psi(S)| \leq 2^{-k/4}.
	\label{eqn:mask-max-fourier-coeff-psi}
\end{align}
In view of (\ref{eqn:mask-psi-orthog-restated}) and
(\ref{eqn:mask-max-fourier-coeff-psi}), Theorem~\ref{thm:pattern-spectrum}
yields
\begin{align}
\|\Psi\|
\leq  2^{-\delta\sqrt k/2} \, 2^{-k/2}.
\label{eqn:Psi-norm}
\end{align}
Now, put 
\begin{align*}
   M=\g(z)
     \left[\g\left(z\oplus \bigoplus_{i=1}^{k/2}\{x_iy_{2i-1}e_{2i-1}
     \oplus\overline{x_i}y_{2i}e_{2i}\}\right)
	  \right]_{x\in\zoo^{k/2},\, y\in\Vcal(k,k/4)},
\end{align*}
where we identify each $y\in\Vcal(k,k/4)$ in the
natural way with a string in 
$\zoo^k.$ Observe that
\begin{align*}
   M= \Big[f_{V,w}(x|_V\oplus
   w)\Big]_{x\in\zoo^{k/2},(V,w)\in\Vcal(k/2,k/4)\times\zoo^{k/4}}
\end{align*}
for some functions $f_{V,w}\in \Fcal.$ This representation makes it
clear, in view of (\ref{eqn:mask-psi-correl-restated}), that
\begin{align}
\langle \Psi, M\rangle > \frac13.    \label{eqn:Psi-M-correl}
\end{align}
By (\ref{eqn:Psi-bounded}), (\ref{eqn:Psi-norm}),
(\ref{eqn:Psi-M-correl}) and the generalized discrepancy method
(Theorem~\ref{thm:discrepancy-method}), we have
$Q^*_{1/10}(M)\geq\Omega(\sqrt k).$ It remains to note that $M$ is
a submatrix of $\g(z)G,$ so that $Q^*_{1/10}(G)\geq Q^*_{1/10}(M).$
}
$\Box$

We will also need the following equivalent formulation of
Theorem~\ref{thm:main-technical-pattern-theorem}, for disjunctions
instead of conjunctions.

\begin{corollary}
  \label{cor:main-technical-pattern-theorem}
Let $\g\colon\zoon\to\moo$ be a function such that 
$\g(z\oplus e_1)=
 \g(z\oplus e_2)=
 \cdots=
 \g(z\oplus e_k)
 \ne \g(z)$ 
for some $z\in\zoon$ with $z_1=\cdots=z_k=1.$ Then the matrix
$G=[\g(x\vee y)]_{x,y\in\zoon}$ satisfies
\begin{align*}
	Q^*_{1/3}(G) \geq \Omega(\sqrt k).
\end{align*}
\end{corollary}

\begin{proof}
{
Put $\tilde \g = \g_{(1,\dots,1)}$ and 
$\tilde z = (1,\dots,1)\oplus z.$ Then 
$\tilde z_1=\cdots=\tilde z_k=0$ and 
$\tilde \g(\tilde z\oplus e_1)=
 \tilde \g(\tilde z\oplus e_2)=
 \cdots=
 \tilde \g(\tilde z\oplus e_k)
 \ne \tilde \g(\tilde z).$ By
Theorem~\ref{thm:main-technical-pattern-theorem}, the matrix $\tilde
G=[\tilde \g(x\wedge y)]_{x,y\in\zoon}$ satisfies $Q^*_{1/3}(\tilde G)\geq
\Omega(\sqrt k).$ It remains to note that $G$ and $\tilde G$ are
identical, up to permutations of rows and columns.
}
\end{proof}

\vspace{\proofskip}

We point out another simple corollary to 
Theorem~\ref{thm:main-technical-pattern-theorem}.

\begin{corollary}
\label{cor:new-quantum-lb}
Let $f\colon\zoon\to\moo$ be given. Then for some $z\in\zoon,$
the matrix $F=[f_z(x\wedge y)]_{x,y} = [f(\dots,(x_i\wedge y_i)\oplus
z_i,\dots)]_{x,y}$ obeys
\begin{align*}
Q^*_{1/3}(F) = \Omega(\sqrt{\bs(f)}).
\end{align*}
\end{corollary}

\begin{proof}
{
Put $k=\bs(f)$ and fix $z\in\zoon$ such that $\zbs(f_z)=k.$ 
By an argument analogous to Lemma~\ref{lem:projection}, one
obtains a function $\g\colon\zoon\to\moo$ such that
$
\g(e_1)=\g(e_2)=\cdots=\g(e_{k})\ne \g(0)
$
 and
$
\g(x) \equiv f_z(\xi_1,\xi_2,\dots,\xi_n)
$
for some symbols $\xi_1,\xi_2,\dots,\xi_n\in\{x_1,x_2,\dots,x_n,0,1\}.$
Then Theorem~\ref{thm:main-technical-pattern-theorem}
implies that the matrix $G=[\g(x\wedge y)]_{x,y\in\zoon}$ satisfies
$Q^*_{1/3}(G)\geq\Omega(\sqrt k).$ On the other hand, 
$Q^*_{1/3}(F)\geq
Q^*_{1/3}(G)$ by construction.
}
\end{proof}

\section{Quantum-Classical Equivalence} \label{sec:qce}

We now combine the combinatorial and analytic development of the
previous sections to obtain our main results. 
We start
by proving relevant lower bounds against quantum protocols.

\begin{theorem}
   \label{thm:quantum-and-or}
Let $f\colon\zoon\to\moo$ be given. Put $F_1=[f(x\wedge y)]_{x,y}$
and $F_2=[f(x\vee y)]_{x,y},$ where the row and column indices range over
$\zoon.$ Then
\begin{align*}
   \max\{ Q^*_{1/3}(F_1), Q^*_{1/3}(F_2)\} = \Omega(\bs(f)^{1/4}).
\end{align*}
\end{theorem}

\begin{proof}
{
By Lemma~\ref{lem:projection}, there exists a function $\g\colon\zoon\to\moo$
such that 
\begin{align}
\s(\g)\geq \Omega(\sqrt{\bs(f)})
\label{eqn:g-sensitivity-restated}
\end{align}
and
\begin{align}
\g(x) \equiv f(x_{i_1},x_{i_2},\dots,x_{i_n})
\label{eqn:g-subfunction-restated}
\end{align}
for some $i_1,i_2,\dots,i_n\in\{1,2,\dots,n\}.$ By renumbering the
variables if necessary, we see that at least one of the following
statements must hold:
\begin{enumerate}
\item[(1)] $\g(z\oplus e_1)=\g(z\oplus e_2)=\cdots=\g(z\oplus e_{\lceil
\s(\g)/2\rceil})\ne \g(z)$ for some $z\in\zoon$ with
$z_1=z_2=\cdots=z_{\lceil \s(\g)/2\rceil}=0$;
\item[(2)] $\g(z\oplus e_1)=\g(z\oplus e_2)=\cdots=\g(z\oplus e_{\lceil
\s(\g)/2\rceil})\ne \g(z)$ for some $z\in\zoon$ with
$z_1=z_2=\cdots=z_{\lceil \s(\g)/2\rceil}=1.$
\end{enumerate}
In the former case, Theorem~\ref{thm:main-technical-pattern-theorem}
implies that the matrix $G_1=[\g(x\wedge y)]_{x,y\in\zoon}$ satisfies
$Q^*_{1/3}(G_1)\geq\Omega(\sqrt{\s(\g)}),$ whence $Q^*_{1/3}(F_1)\geq
Q^*_{1/3}(G_1)\geq\Omega(\bs(f)^{1/4})$ in view
of~(\ref{eqn:g-sensitivity-restated}) and
(\ref{eqn:g-subfunction-restated}).

In the latter case, Corollary~\ref{cor:main-technical-pattern-theorem}
implies that \mbox{$G_2=[\g(x\vee y)]_{x,y\in\zoon}$} satisfies
$Q^*_{1/3}(G_2)\geq\Omega(\sqrt{\s(\g)}),$ whence $Q^*_{1/3}(F_2)\geq
Q^*_{1/3}(G_2)\geq\Omega(\bs(f)^{1/4})$ in view
of~(\ref{eqn:g-sensitivity-restated}) and
(\ref{eqn:g-subfunction-restated}).
}
\end{proof}

\vspace{\proofskip}

Having obtained the desired lower bounds on quantum communication, we
now turn to classical protocols. The bound that we seek here follows
easily from the work of Buhrman et al.~\cite{bcw98quantum} and Beals
et al.~\cite{beals-et-al01quantum-by-polynomials}. Related observations
have been used in a number of recent papers in the area~\cite{razborov03quantum,
sherstov07quantum, shi-zhu07block-composed}.

\begin{theorem}[Classical upper bound;
cf.~\cite{bcw98quantum,beals-et-al01quantum-by-polynomials}]
   \label{thm:classical-and-or}
Let $f\colon\zoon\to\moo$ be given. Put $F_1=[f(x\wedge y)]_{x,y}$
and $F_2=[f(x\vee y)]_{x,y},$ where the row and column indices range over
$\zoon.$ Then
\begin{align*}
   \max\{ D(F_1), D(F_2)\} \leq 2\dt(f)\leq 2\bs(f)^3.
\end{align*}
\end{theorem}

\indent {\bf Proof (adapted from~\cite{bcw98quantum,
beals-et-al01quantum-by-polynomials}).}
{
The second inequality follows immediately by 
Theorem~\ref{thm:beals}, so we will focus on the first. 
Fix an optimal-depth decision tree for $f.$ The protocol for $F_1$
is as follows.  On input $x$ and $y,$ Alice and Bob start at the
top node of the tree, read its label $i,$ and exchange the
two bits $x_i$ and $y_i.$ This allows them to compute $x_i\wedge
y_i$ and to determine which branch to take next. The process repeats
at the new node and so on, until the parties have reached a leaf
node. Since the longest root-to-leaf path has length $\dt(f),$ the
claim follows.  The proof for $F_2$ is entirely analogous.
}
$\Box$

Theorems~\ref{thm:quantum-and-or} and~\ref{thm:classical-and-or}
immediately imply our main result on quantum-classical equivalence,
stated above as Theorem~\ref{thm:main}.

\section{Masked Problems and the Log-Rank Conjecture} \label{sec:logrank}

As we showed in the previous section, the communication problem of
computing \mbox{$f(x\wedge y)$} and \mbox{$f(x\vee y)$} has
polynomially related quantum and classical complexities. Here, we will
see that this communication problem additionally satisfies the
\emph{log-rank conjecture} of 
Lov{\'a}sz and Saks~\cite{lovasz-saks88logrank}.

The log-rank conjecture states that the deterministic
communication complexity of every Boolean matrix $F$ satisfies
$D(F)\leq (\log_2 \rk F)^c + c$ for some absolute constant $c>0.$ By
Theorem~\ref{thm:mehlhorn-schmidt}, this is equivalent to saying
that $D(F)$ is polynomially related to $\log_2 \rk F.$ The development in
this section is based on the following result of Buhrman and de
Wolf~\cite{buhrman-dewolf01polynomials}, who studied the special case of symmetric
functions $f$ in the same context.

\begin{theorem}[Buhrman and de Wolf]
Let $f\colon\zoon\to\Re$ be a given function. 
Put $M=[f(x\wedge y)]_{x,y},$ where the
row and column indices range over $\zoon.$ Then
\[ \rk M = \mon(f).\]
\label{thm:rk-mon}
\end{theorem}

Our first observation is as follows.

\begin{lemma}
Let $f\colon\zoon\to\Re$ be a given function, where $f\not\equiv 0$ and
$d=\deg(f).$ Then for some $z\in\zoon,$
\begin{align*}
\mon(f_z) \geq \left(\frac32\right)^d.
\end{align*}
\label{lem:hard-shift}
\end{lemma}

\begin{proof}
{
The proof is by induction on $d.$ The base case $d=0$ holds since
\mbox{$f\not\equiv 0.$} Assume that the claim holds for all $f$ of degree
$d-1.$ By renumbering the variables if necessary, we have
$f(x) = x_1 p(x_2,\dots,x_n) + q(x_2,\dots,x_n)$
for some polynomial $p$ of degree $d-1.$ The inductive assumption
guarantees the existence of $u\in\zoo^{n-1}$ such that
$\mon(p_u)\geq (3/2)^{d-1}.$ Note that 
$\mon(f_{(0,u)}) = \mon(p_u) + \mon(q_u)$ and
$\mon(f_{(1,u)}) \geq \mon(p_u) + \lvert\mon(q_u)-\mon(p_u)\rvert.$
Therefore,
\begin{align*}
\max\{ \mon(f_{(0,u)}),\mon(f_{(1,u)})\}
\geq \frac32 \mon(p_u) \geq \left(\frac 32\right)^d,
\end{align*}
as desired.
}
\end{proof}

\vspace{\proofskip}

We will also need the following technical lemma.

\begin{lemma}
Let $f\colon\zoon\to\Re$ be given. Fix an index $i=1,2,\dots,n.$
Define 
\[ \tilde f(x_1,\dots,x_{i-1},x_{i+1},\dots,x_n)
=f(x_1,\dots,x_{i-1}, 0,x_{i+1},\dots,x_n).\]
Then
\[ \max\{ \mon(\tilde f),\; \mon(f_{e_i})\} \geq \frac12 \mon(f). \]
\label{lem:restrictions}
\end{lemma}

\begin{proof}
{
Write
\[f(x) = x_i p(x_1,\dots,x_{i-1},x_{i+1},\dots,x_n) +
   \tilde f(x_1,\dots,x_{i-1},x_{i+1},\dots,x_n).\]
It is clear by inspection
that $\mon(f_{e_i}) \geq \mon(p).$ Thus, we have 
$\mon(\tilde f) + \mon(f_{e_i})
    \geq \mon(\tilde f)+ \mon(p) 
    =\mon(f),
$ as desired.
}
\end{proof}

\vspace{\proofskip}

At last, we arrive at the main result of this section.

\begin{theorem}
Let $f\colon\zoon\to\moo$ be given, $d=\deg(f).$
Put $F_1=[f(x\wedge y)]_{x,y}$ and $F_2=[f(x\vee y)]_{x,y},$ where
the row and column indices range over $\zoon.$ Then
\begin{align}
   \max\{\rk F_1,\, \rk F_2\} \geq \left(\frac3{2\sqrt 2}\right)^d 
  \geq 1.06^d. 
  \label{eqn:rk-lower-bound}
\end{align}
In particular, the communication problem of computing, on input
$x,y\in\zoon,$ both of the quantities $f(x\wedge y)$ and $f(x\vee y),$
satisfies the log-rank conjecture.
  \label{thm:rk-lower-bound}
\end{theorem}

\begin{proof}
{
To see how the last statement follows from the lower bound
(\ref{eqn:rk-lower-bound}), note that $\max\{D(F_1),D(F_2)\}\leq
2\dt(f)$ by Theorem~\ref{thm:classical-and-or} and 
$\dt(f)\leq O(\deg(f)^3)$ by Theorem~\ref{thm:dt-deg}.  In
the remainder of the proof, we focus on (\ref{eqn:rk-lower-bound})
alone.

We assume that $d\geq1,$ the claim being trivial otherwise.
By renumbering the variables if necessary, we may write
\[ f(x) = \alpha x_1x_2\cdots x_d + \sum_{S\ne\{1,\dots,d\}}
\alpha_S \prod_{i\in S} x_i, \]
where $\alpha\ne 0.$  Define $\g(x_1,\dots,x_d)=f(x_1,\dots,x_d,0,\dots,0).$
Then $\g$ is a nonzero polynomial of degree $d,$ and
Lemma~\ref{lem:hard-shift} yields a vector $z\in\zoo^d$ such that
\[ 
  \mon(\g_z)\geq \left(\frac32\right)^d.
\]
By renumbering the variables if necessary, we may assume that
$z=0^t1^{d-t}.$ We complete the proof by analyzing the cases $t\leq
d/2$ and $t>d/2.$

Suppose first that $t\leq d/2.$ Let $\mathcal F$ be the set whose elements
are the identity function on $\zoo$ and the constant-one function
on $\zoo.$ Lemma~\ref{lem:restrictions} provides functions
$\phi_1,\dots,\phi_t\in\mathcal{F}$ such that the polynomial 
$h(x_1,\dots,x_d) 
 = \g_{1^d}(\phi_1(x_1),\dots,\phi_t(x_t),x_{t+1},\dots,x_{d})$
features at least $2^{-t} \mon(\g_z)\geq (3/\{2\sqrt 2\})^d$
monomials. By Theorem~\ref{thm:rk-mon}, the matrix $H=[h(x\wedge
y)]_{x,y\in\zoo^d}$ has rank at least $(3/\{2\sqrt 2\})^d.$ Since $H$
is a submatrix of $F_2,$ the theorem holds in this case.

The case $t>d/2$ is entirely symmetric, with $F_1$ playing the role of
$F_2.$
}
\end{proof}

\begin{remark}
{\rm 
By the results of Buhrman and de Wolf~\cite{buhrman-dewolf01polynomials},
Theorem~\ref{thm:rk-lower-bound} alone would suffice to obtain a
polynomial relationship between classical and quantum communication
complexity in the \emph{exact} model. However, for our main result
we need a polynomial relationship in the \emph{bounded-error} model,
which requires the full development
of Sections~\ref{sec:combinatorial}--\ref{sec:qce}.
}
\end{remark}

\section{Results for Composed Functions} \label{sec:composed}

Up to this point, we have focused on the communication problem of
computing \mbox{$f(x\wedge y)$} and $f(x\vee y).$ Here we point out 
that our results on quantum-classical equivalence and the log-rank
conjecture immediately apply to a broader class of communication problems.
Specifically, we will consider compositions of the form
$f(g_1(x^{(1)},y^{(1)}),\dots,g_n(x^{(n)},y^{(n)})),$ where one has a combining
function $f\colon \zoon\to\moo$ that receives input from intermediate
functions $\g_i\colon X_i\times Y_i\to\zoo,$ $i=1,2,\dots,n.$ We
will show that under natural assumptions on $\g_1,\dots,\g_n,$ this
composed function will have polynomially related quantum and classical
bounded-error complexities and will satisfy the log-rank conjecture.
To simplify notation, we will henceforth 
abbreviate $f(g_1(x^{(1)},y^{(1)}),\dots,g_n(x^{(n)},y^{(n)}))$
to $f(\dots,g_i(x^{(i)},y^{(i)}),\dots).$

\begin{theorem}
Let $f\colon\zoon\to\moo$ be a given function. Fix functions $\g_i\colon X_i\times
Y_i\to\zoo,$ for $i=1,2,\dots,n.$ Assume that for each $i,$ the
matrix $[\g_i(x^{(i)},y^{(i)})]_{x^{(i)}\in X_i,y^{(i)}\in Y_i}$
contains the following submatrices
\begin{align}
\begin{bmatrix}
1&0\\
0&0
\end{bmatrix},\qquad
\begin{bmatrix}
0&1\\
1&1
\end{bmatrix},
\label{eqn:submatrices}
\end{align}
up to permutations of rows and columns.
Put $F=[f(\dots,\g_i(x^{(i)},y^{(i)}),\dots)].$ 
Assume that for some constant $\alpha>0,$
\begin{align}
Q^*_{1/3}(\g_i)\geq R_{1/3}(\g_i)^\alpha, \qquad i=1,2,\dots,n.
\label{eqn:g_i-equivalence}
\end{align}
Then for some constant $\beta=\beta(\alpha)>0,$
\begin{align*}
R_{1/3}(F) \geq Q_{1/3}^*(F)\geq R_{1/3}(F)^\beta.
\end{align*}
   \label{thm:general-equivalence}
\end{theorem}

\begin{proof}
{
Without loss of generality, we may assume that $f$ depends on all of its
$n$ inputs (otherwise, disregard any irrelevant inputs from among
$\g_1,\dots,\g_n$ in the analysis below). In particular, we have
\begin{align}
Q_{1/3}^*(F) \geq  Q_{1/3}^*(\g_i),\qquad i=1,2,\dots,n.
\label{eqn:Q-projection}
\end{align}
Since each $\g_i$ contains the two-variable functions AND and OR 
as subfunctions, Corollary~\ref{cor:new-quantum-lb} shows that
\begin{align}
Q_{1/3}^*(F) \geq \Omega(\sqrt{\bs(f)}).
\label{eqn:F-quantum-lb}
\end{align}
Letting $d=\dt(f),$ we claim that 
\begin{align}
R_{1/3}(F) \leq  O(d\log d) \max_{i=1,\dots,n}\{ R_{1/3}(\g_i)\}.
\label{eqn:F-classical-ub}
\end{align}
The proof of this bound is closely analogous to that of
Theorem~\ref{thm:classical-and-or}. Namely, Alice and Bob evaluate
a depth-$d$ decision tree for $f.$ When a tree node calls for the
$i$th variable, the parties run an optimal randomized protocol for
$\g_i$ with error probability $\frac1{3d},$ which requires at most
$O(R_{1/3}(\g_i)\log d)$ bits of communication. Since all root-to-leaf paths
have length at most $d,$ the final answer will be correct with
probability at least $2/3.$

In view of Theorem~\ref{thm:beals}, the sought polynomial relationship 
between $R_{1/3}(F)$ and $Q_{1/3}^*(F)$ follows from
(\ref{eqn:g_i-equivalence})--(\ref{eqn:F-classical-ub}).
}
\end{proof}

\vspace{\proofskip}

We now record an analogous result for the log-rank conjecture.

\begin{theorem}
Let $f\colon\zoon\to\moo$ be a given function. Fix functions $\g_i\colon X_i\times
Y_i\to\zoo,$ for $i=1,2,\dots,n.$ Assume that for each $i,$ the
matrix $[\g_i(x^{(i)},y^{(i)})]_{x^{(i)}\in X_i,y^{(i)}\in Y_i}$
contains \textup{(\ref{eqn:submatrices})} as submatrices, 
up to permutations of rows and columns.
Assume that for some constant $c>0,$
\begin{align}
D(\g_i)\leq (\log_2 \rk G_i)^c + c, \qquad i=1,2,\dots,n,
\label{eqn:g_i-logrank}
\end{align}
where
$G_i=[(-1)^{\g_i(x^{(i)},y^{(i)})}]_{x^{(i)}\in X_i,y^{(i)}\in Y_i}.$
Then 
the matrix $F=[f(\dots,\g_i(x^{(i)},y^{(i)}),\dots)]$ 
obeys
\begin{align*}
D(F)\leq (\log_2 \rk F)^C + C
\end{align*}
for some constant $C=C(c)>0.$
In particular, $F$ satisfies the log-rank conjecture.
   \label{thm:general-logrank}
\end{theorem}

\begin{proof}
{
Without loss of generality, we may assume that $f$ depends on all of its
$n$ inputs (otherwise, disregard any irrelevant inputs from among
$\g_1,\dots,\g_n$ in the analysis below). In particular, we have
\begin{align}
\rk F \geq \rk G_i, \qquad i=1,2,\dots,n.
\label{eqn:F-logrank-projection}
\end{align}
Since each $\g_i$ contains the two-variable functions AND and OR 
as subfunctions, Theorem~\ref{thm:rk-lower-bound} shows that
\begin{align}
\rk F \geq \left(\frac3{2\sqrt 2}\right)^{\deg(f)}.
\label{eqn:F-and-or}
\end{align}
Finally, we claim that 
\begin{align}
D(F) \leq  2\dt(f) \max_{i=1,\dots,n}\{ D(\g_i)\}.
\label{eqn:F-deterministic-ub}
\end{align}
The proof of this bound is closely analogous to that of
Theorem~\ref{thm:classical-and-or}. Namely, Alice and Bob evaluate
an optimal-depth decision tree for $f.$ When a tree node calls for the
$i$th variable, the parties run an optimal deterministic protocol for
$\g_i.$ 

In view of (\ref{eqn:g_i-logrank})--(\ref{eqn:F-deterministic-ub}) and
Theorem~\ref{thm:dt-deg}, the proof is complete.
}
\end{proof}

The key property of $\g_1,\dots,\g_n$ that we have used in this
section is that their communication matrices contain
(\ref{eqn:submatrices}) as submatrices. We close this section by
observing that this property almost always holds.  More precisely,
we show that matrices that do not contain the submatrices
(\ref{eqn:submatrices}) have a very restricted structure.

\begin{theorem}
\label{thm:01}
A matrix $G\in\zoo^{N\times M}$ does not contain 
\begin{align*}
A = \begin{bmatrix}
0&1\\
1&1
\end{bmatrix}
\end{align*}
as a submatrix if and only if\; $G=0,$ \;$G=J,$ \;or
\begin{align}
G' \sim 
\begin{bmatrix}
\begin{aligned}
\begin{matrix}
J_1\\
&J_2\\
&&J_3\\
\end{matrix}
\end{aligned} & \scalebox{1.7}{$0$}\\
\begin{matrix}
\\
\scalebox{1.7}{$0$~}
\end{matrix} & \begin{aligned}
\begin{matrix}
\ddots\\
&J_k\\
\end{matrix}
\end{aligned}
\end{bmatrix}, \qquad
\label{eqn:G-characterization}
\end{align}
where: $G'$ is the result of deleting any columns and rows in $G$
that consist entirely of zeroes; $J,J_1,J_2,\dots,J_k$ are
all-$1$ matrices of appropriate dimensions; and $\sim$ denotes
equality up to permutations of rows and columns.
\end{theorem}

\begin{proof}
{
The ``if'' part is clear.  We will prove the other direction by 
induction on the number of columns, $M.$ The base case is trivial.
For the inductive step, let $G\ne 0$ be a given matrix. Let $J_1$
be a maximal submatrix of $G$ with all entries equal to $1.$ Then
\begin{align*}
G \sim 
\begin{bmatrix}
J_1 & Z_1\\
Z_2 & H
\end{bmatrix}
\end{align*}
for suitable matrices $Z_1,Z_2,$ and $H,$ possibly empty. By the
maximality of $J_1$ and the fact that $G$ does not contain $A$ as
a submatrix, it follows that either $Z_1$ is empty or $Z_1=0.$
Likewise for $Z_2.$ By the inductive hypothesis for $H,$ the proof
is complete.
}
\end{proof}

By reversing the roles of $0$ and $1,$ one obtains from 
Theorem~\ref{thm:01} an analogous characterization of all matrices
$G=\zoo^{N\times M}$ that do not contain 
\begin{align*}
\begin{bmatrix}
1&0\\
0&0
\end{bmatrix}
\end{align*}
as a submatrix. 
%

\begin{remark}
{\rm 
The communication complexity of a Boolean matrix remains unaffected if
one modifies it to retain only one copy of each column, removing any
duplicates. An analogous statement holds for the rows.
In light of Theorem~\ref{thm:01}, this means
that there are only four types of intermediate functions $\g$ for
which our composition results (Theorem~\ref{thm:general-equivalence}
and~\ref{thm:general-logrank}) fail.  These are the functions $\g$ with
matrix representations
\begin{align}
I, \qquad 
\begin{bmatrix}
I\\
& 0
\end{bmatrix},
\qquad
\begin{bmatrix}
I\\
0 
\end{bmatrix},
\qquad
\begin{bmatrix}
I & 0
\end{bmatrix},
\label{eqn:four-matrix-types}
\end{align}
and their negations, where $I$ is the identity matrix.
The reason that Theorems~\ref{thm:general-equivalence}
and~\ref{thm:general-logrank} fail for such $\g$ is that the
underlying quantum lower bound in terms of block sensitivity of
the combining function $f$ is no longer valid. For example, the
first matrix type, $I,$ corresponds to letting $\g$ be the
equality function. Now, the conjunction of $n$ equality functions
is still an equality function, and its communication complexity
is $O(1)$ both in the randomized and quantum
models~\cite{ccbook}, which is much less than a hypothetical
lower bound of $\Omega(\sqrt n)$ that one would expect from the
block sensitivity of $f=\AND_n.$ The same $O(1)$ upper
bound holds for a conjunction of arbitrarily many functions $\g$
of the second, third, and fourth type. 
}
\end{remark}

%

\nonumsection{Acknowledgments}
\noindent
The author would like to thank 
Dima Gavinsky, 
Adam Klivans,
Sasha Razborov, 
and the anonymous reviewers for their useful
comments on a preliminary version of this paper.

\nonumsection{References}
\bibliographystyle{abbrv}
\bibliography{refs}

\end{document}